\documentclass[runningheads]{llncs}

\usepackage{amsfonts}
\usepackage{mathrsfs}
\usepackage{amssymb}
\usepackage{amsmath}
\usepackage{cite}
\usepackage{xcolor}
\usepackage{graphicx}
\usepackage{algorithm}
\usepackage{algorithmicx}
\usepackage{authblk}
\usepackage{xspace}
\usepackage{hyperref}
\hypersetup{
    unicode=false,          
    colorlinks=true,        
    linkcolor=blue,          
    citecolor=green,        
    filecolor=magenta,      
    urlcolor=cyan           
}
\usepackage{cleveref}
\usepackage{thmtools,thm-restate}

\usepackage[title]{appendix}

\usepackage[noend]{algpseudocode}

\newcommand{\smr}[0]{Leader-Aware SMR\xspace}
\newcommand{\code}[1]{\ensuremath{\mathtt{#1}}\xspace}

\newcommand{\namedref}[2]{\hyperref[#2]{#1~\ref*{#2}}}

\newcommand{\equationref}[1]{\hyperref[#1]{(\ref*{#1})}}
\newcommand{\lemmaref}[1]{\namedref{Lemma}{#1}}
\newcommand{\lineref}[1]{\namedref{line}{#1}}
 
\newcommand{\algorithmref}[1]{\namedref{Algorithm}{#1}}

\newcommand{\propertyref}[1]{\namedref{Property}{#1}} 
\newcommand{\figureref}[1]{\namedref{Figure}{#1}} 
\newcommand{\sectionref}[1]{\namedref{Section}{#1}} 
 
\newcommand{\appendixref}[1]{\namedref{Appendix}{#1}} 

\newcommand{\carousel}[0]{Carousel\xspace}
\newcommand{\lr}[1]{\langle #1 \rangle}
\newcommand{\shir}[1]{[\textbf{\textcolor{red}{Shir:}} {\footnotesize \textcolor{red}{#1}]}}
\newcommand{\sasha}[1]{[\textbf{\textcolor{blue}{Sasha:}} {\footnotesize \textcolor{blue}{#1}]}}

\newcommand{\lef}[1]{[\textbf{\textcolor{brown}{Lefteris:}} {\footnotesize \textcolor{brown}{#1}]}}
\newcommand{\rati}[1]{[\textbf{\textcolor{green}{Rati:}} {\footnotesize \textcolor{green}{#1}]}}
\newcommand{\alberto}[1]{[\textbf{\textcolor{purple}{Alberto:}} {\footnotesize \textcolor{purple}{#1}]}}
\newcommand{\dahlia}[1]{[\textbf{\textcolor{orange}{Dahlia:}} {\footnotesize \textcolor{orange}{#1}]}}
\long\def\com#1{}

\algdef{SE}[UPON]{Upon}{EndUpon}[1]{\textbf{upon
}\ #1\ \algorithmicdo}{\algorithmicend\ \textbf{}}%
\algtext*{EndUpon}

\def\cameraReady{} 

\begin{document}

\title{Be Aware of Your Leaders}

\ifdefined\cameraReady
\author{
Shir Cohen\inst{1,2} \and
Rati Gelashvili\inst{1} \and
Lefteris Kokoris Kogias\inst{1,3} \and
Zekun Li\inst{1}
\and Dahlia Malkhi\inst{1} \and 
Alberto Sonnino\inst{1} \and 
Alexander Spiegelman\inst{1}
}
\institute{Novi Research \and Technion \and IST Austria}
\else
\author{}
\institute{}
\fi

\authorrunning{S. Cohen et al.}

\maketitle              
\begin{abstract}
Advances in blockchains have influenced the State-Machine-Replication (SMR) world and many state-of-the-art blockchain-SMR solutions are based on two pillars: \emph{Chaining} and \emph{Leader-rotation}.
A predetermined round-robin mechanism used for Leader-rotation, however,
has an undesirable behavior: crashed parties become designated leaders infinitely often, slowing down overall system performance.  
%
%
%
In this paper, we provide a new \smr framework that, among other desirable properties, formalizes a  \emph{Leader-utilization} requirement that bounds the number of rounds whose leaders are faulty in crash-only executions. 
\vspace{0.33em}

%
%
%
We introduce \carousel, a novel, reputation-based Leader-rotation solution to achieve \smr. The challenge in adaptive Leader-rotation is that it cannot rely on consensus to determine a leader, since consensus itself needs a leader. \carousel uses the available on-chain information to determine a leader locally and achieves Liveness despite this difficulty. A HotStuff implementation fitted with \carousel demonstrates drastic performance improvements: it increases throughput over 2x in faultless settings and provided a 20x throughput increase and 5x latency reduction in the presence of faults.

\end{abstract}

\section{Introduction}

Recently, Byzantine agreement protocols in the eventually synchronous model such as Tendermint~\cite{buchman2016tendermint}, Casper FFG~\cite{buterincasper}, and HotStuff~\cite{yin2019hotstuff}, brought two important concepts from the world of blockchains to the traditional State Machine Replication (SMR)~\cite{lamport2019time} settings, \emph{Leader-rotation} and \emph{Chaining}. More specifically, these algorithms operate by designating one party as \textit{leader} of each round to propose the next block of transactions that extends a \textit{chained} sequence of blocks.
Both properties depart from the approach used by classical protocols such as PBFT~\cite{castro1999practical}, Multi-Paxos~\cite{lamport2001paxos} and Raft~\cite{ongaro2014search} (the latter two in benign settings). In those solutions, a stable leader operates until it fails and then it is replaced by a new leader.
Agreement is formed on an immutable sequence of indexed (rather than chained) transactions, organized in slots.


Leader-rotation is important in a Byzantine setting, since parties should not trust each other for load sharing, reward management, resisting censoring of submitted transactions, or ordering requests fairly~\cite{kelkar2020order}.
The advantage of Chaining is that it simplifies the leader handover since in the common case the chain eliminates the need for new leaders to catch up with outcomes from previous slots. 


In the permissioned SMR settings~\cite{fabric}, most existing Leader-rotation mechanisms use a round-robin approach to rotate leaders~\cite{DiemBFT,yin2019hotstuff,chan2020streamlet}.
This guarantees that honest parties get a chance to be leaders infinitely often, which is sufficient to drive progress and satisfy \emph{Chain-quality}~\cite{garay2015bitcoin}.
Roughly speaking, the latter stipulates that the number of blocks committed to the chain by honest parties is proportional to the honest nodes' percentage. 
The drawback of such a mechanism is that it does not bound the number of faulty parties which are designated as leaders during an execution.
This has a negative effect on latency even in crash-only executions, as each crashed leader delays progress. 
Similarly to XFT~\cite{liu2016xft}, we seek to improve the performance in such executions, while we also maintain Chain-quality to thwart Byzantine attacks.



In this paper, we propose a leader-rotation mechanism, 
\carousel, 
that enjoys both worlds.
\carousel satisfies non-zero Chain-quality, and at the same time, bounds the number of faulty leaders in crash-only executions after the global stabilization time (GST), a property we call \emph{Leader-utilization}.
The \carousel algorithm leverages Chaining to execute purely locally using information available on the chain, avoiding any extra communication.
To capture all requirements, we formalize  a \emph{\smr} problem model, which alongside Agreement, Liveness and Chain-quality, also requires Leader-utilization. We prove that \carousel satisfied the \smr requirements.


The high-level idea to satisfy Leader-utilization is to track active parties via the records of their participation (e.g. signatures) at the committed chain prefix 
and elect leaders among them.
However, if done naively, the adversary can exploit this mechanism to violate Liveness or Chain-quality.
The challenge is that there is no consensus on a committed prefix to determine a leader, since consensus itself needs a leader. Diverging local views on committed prefixes may be effectuated, for instance, by having a Byzantine leader reveal an updated head of the chain to a subset of the honest parties. Hence, \carousel may not have agreement on the leaders of some rounds, but nevertheless guarantees Liveness and Leader-utilization after GST. 



To focus on our leader-rotation mechanism, we abstract away all other SMR components by defining an SMR framework.
Similarly to~\cite{spiegelman2021ace}, we capture the logic and properties of forming and certifying blocks of transactions in each round in a \emph{Leader-based round (LBR)} abstraction, and rely on a Pacemaker abstraction~\cite{naor2020expected,naor2019cogsworth,bravo2020making} for round synchronization.
We prove that when instantiated into this framework, \carousel yields a \smr protocol.
Specifically, we show (1) for Leader-utilization: at most $O(f^2)$ faulty leader may be elected in crash-only executions (after GST); and (2) for Chain-quality: one out of $O(f)$ blocks is authored by an honest party in the worst-case. 

We provide an implementation of \carousel in a HotStuff-based system and an evaluation that demonstrates a significant performance improvement. Specifically, we get over 2x throughput increase in faultless settings, and 20x throughput increase and 5x latency reduction in the presence of faults.
Our mechanism is adopted in the most recent version of DiemBFT~\cite{DiemBFT}, a deployed HotStuff-based system.

\section{Model and Problem Definition}

We consider a  message-passing model with a set of $n$ parties $\Pi=\{p_1,\dots,p_n\}$, out of which $f<\frac{n}{3}$ are subject to failures. A party is \emph{crashed} if it halts prematurely at some point during an execution. If it deviates from the protocol it is \emph{Byzantine}. An \emph{honest} party never crashes or becomes Byzantine.
We say that an execution is \emph{crash-only} if there are no Byzantine failures therein.

For the theoretical analysis we assume an eventually synchronous communication model~\cite{dwork1988consensus} in which there is a global stabilization time (GST) after which the network becomes synchronous.
That is, before GST the network is completely asynchronous, while after GST messages arrive within a known bounded time, denoted as $\delta$.




As we later describe, we abstract away much of the SMR implementation details by defining and using primitives. 
Therefore, our Leader-rotation solution is model agnostic and the adversarial model depends on the implementation choices for those primitives. 




\subsection{Leader-Aware SMR}
In this section we introduce some notation and then define the \smr problem.
Roughly speaking, \smr captures the desired properties of the Leader-rotation mechanism in SMR protocols that are leader-based.

An SMR protocol consists of a set of parties aiming to maintain a growing chain of \emph{blocks}.
Parties participate in a sequence of rounds, attempting to form a block per round.
In \smr, each round is driven by a leader. We capture these rounds via the Leader-based round (LBR) abstraction defined later. 

A block consists of transactions and the following meta-data:
\begin{itemize}
    \item A (cryptographic) link to a \emph{parent} block. Thus, each block implicitly defines a chain to the genesis block.
    \item A round number in which the block was formed.
    \item The author id of the party that created the block. 
    \item A certificate that (cryptographically) proves that $2f+1$ parties endorsed the block in the given round and with the given author.
    We assume that it is possible to obtain the set of $2f+1$ endorsing parties\footnote{This can be achieved by multi-signature schemes which are practically as efficient as threshold signatures~\cite{multi-signatures}.}.
\end{itemize}
Note that having a round number and the author id as a part of the block is not strictly necessary, but they facilitate formalization of properties and analysis.
For example, an \emph{honest block} is defined as a block authored by an honest party and a \emph{Byzantine block} is a block authored by a Byzantine party.




We assume a predicate $\code{certified}(B,r)\in \{true, false\}$ that locally checks whether the block has a valid certificate, i.e. it has $2f+1$ endorsements for round $r$.
If $\code{certified}(B,r)=true$ we say that $B$ is a \emph{certified} block of round $r$. When clear from context, we say that $B$ is \emph{certified} without explicitly mentioning the round number.


An SMR protocol does not terminate, but rather continues to form blocks. Each block $B$ determines its \emph{implied} chain starting from $B$ to the genesis block via the parent links.
We use notation $B \longrightarrow B'$, saying $B'$ \emph{extends} $B$, if block $B$ is on $B'$'s implied chain.
Honest parties can \emph{commit} blocks in some rounds (but usually not all). A committed block indirectly commits its implied chain.
An SMR protocol must satisfy the following:
\begin{definition}[\smr]
\begin{itemize}
%
    \item \textbf{Liveness:} An unbounded number of blocks are committed by honest parties.

    \item \textbf{Agreement:} If an honest party $p_i$ has committed a block $B$, then for any block $B'$ committed by any honest party $p_j$ either $B \longrightarrow B'$ or $B' \longrightarrow B$. 
    
    \item \textbf{Chain-quality:} 
    For any block $B$ committed by an honest party $p_i$, the proportion of Byzantine blocks on $B$'s implied chain is bounded.
    
    \item \textbf{Leader-Utilization:} In crash-only executions, after GST, the number of rounds $r$ for which no honest party commits a block formed in $r$ is bounded. 
\end{itemize}
\end{definition}
The first two properties are common to SMR protocols.
While most SMR algorithms satisfy the above mentioned Liveness condition, a stronger Liveness property can be defined, requiring that each honest party commits an unbounded number of blocks. This property can be easily be achieved by  an orthogonal forwarding mechanism, where each honest leader that creates a block explicitly sends it to all other parties.
A notion of Chain-quality that bounds the adversarial control over chain contents was first suggested by Garay et al.~\cite{garay2015bitcoin}.
We introduce the Leader-utilization property to capture the quality of the Leader-rotation mechanism in crash-only executions.

\section{\smr: The Framework}
In order to isolate the Leader-rotation problem in \smr protocols, we abstract away the remaining logic into two components.
First, similar to~\cite{spiegelman2021ace,spiegelman2020search} we capture the logic to form and commit blocks by the \emph{Leader-based round (LBR)} abstraction (\sectionref{subsec:lbr}). 
We follow~\cite{bravo2020making,naor2020expected} and capture round synchronization by the Pacemaker abstraction (\sectionref{subsec:pacemaker}).
These two abstractions can be instantiated with known implementations from existing SMR protocols.

In~\sectionref{subsec:le} we define the core API for Leader-rotation and combine it with the above components to construct an SMR protocol.
In~\sectionref{sec:le} we present a Leader-rotation algorithm that can be easily computed based on locally available information and makes the construction a \smr.


\subsection{Leader-based round (LBR)}
\label{subsec:lbr}
The LBR abstraction exposes to each party $p_i$ an API to invoke $LBR(r,\ell)$, where $r \in \mathbb{N}$ is a round number and $\ell$ is the leader of round $r$ according to party $p_i$.
Intuitively, a leader-based round captures an attempt by parties to certify and commit a block formed by the leader\footnote{Existing SMR protocols may have separate rounds (and even leaders) for forming and committing blocks, but this distinction is not relevant for the purposes of the paper and LBR abstraction is defined accordingly.} - which naturally requires sufficiently many parties to agree on the identity of the leader.
We assume that non-Byzantine parties can only endorse a block $B$ with round number $r$ and author $\ell$ by calling $LBR(r,\ell)$.

Every LBR invocation returns within $\Delta_l > c\delta$ time, where $c$ depends on the specific LBR implementation (i.e., each round requires a causal chain of $c$ messages to complete).
That is, $\Delta_l$ captures the inherent timeouts required for eventually synchronous protocols.
We say that round $r$ has $k \leq n$ \emph{LBR-synchronized($\ell$)} invocations if $k$ honest parties invoke $LBR(r,\ell)$ after GST and within $\Delta_l - c \delta$ time of each other with the same party $\ell$\footnote{LBR-synchronized requires that the corresponding execution intervals have a shared intersection lasting $\geq c\delta$ time.
}. 

The return value of an LBR invocation in round $r$ is always a block with a round number $r' \leq r$.
The intention is for LBR invocations to return gradually growing committed chains.
Occasionally, there is no progress, in which case the invocations are allowed to return a committed block whose  round $r'$ is smaller than $r$.  
Formally, the output from LBR satisfies the following properties:
\begin{definition}[LBR]
\label{def:LBR}
\begin{itemize}
    \item \textbf{Endorsement:} For any block $B$ and round $r$, if $\code{certified}(B,r) = true$, then the set of endorsing parties of $B$ contains $2f+1$ parties. 
    
    \item \textbf{Agreement:} If $B$ and $B'$ are certified blocks that are each returned to an honest party from an LBR invocation, then either $B \longrightarrow B'$ or $B' \longrightarrow B$.
    
    
    \item \textbf{Progress:}
    If there are $k \geq 2f+1$ LBR-synchronized($\ell$) invocations at round $r$ and $\ell$ is honest, then they all return a certified $B$ with round number $r$ authored by $\ell$.

    \item \textbf{Blocking:} If an non-Byzantine party $\ell$ never invokes $LBR(r,\ell)$, then no  $LBR(r,\ell)$ invocation may return a certified block formed in round $r$.
    
    \item \textbf{Reputation:}  If a non-Byzantine party $p$ never invokes $LBR$ for round $r$, then any certified block $B$ with round number $r$ does not contain $p$ among its endorsers.

\end{itemize}
\end{definition}
The LBR definition intends to capture just the key properties required for round abstraction in SMR protocols but leaves room for various interesting behavior.
For example, if the progress preconditions are not met at round $r$, then some honest parties may return a block $B$ for round $r$ while others do not. 
Moreover, in this case the adversary can \emph{hide} certified blocks from honest parties and reveal them at any point via the LBR return values.





\subsection{The Pacemaker}
\label{subsec:pacemaker}

The Pacemaker~\cite{naor2019cogsworth,bravo2020making,naor2020expected} component is a commonly used abstraction, which ensures that, after GST, parties are synchronized and participate in the same round long enough to satisfy the LBR progress.
We assume the following:

\begin{definition}[Pacemaker]
\label{def:pacemaker}
The Pacemaker eventually produces \code{new\_round(r)} notifications at honest parties for each round $r$.
Suppose for some round $r$ all \code{new\_round(r)} notifications at non-Byzantine parties occur after GST, the first of which occurs at time $T_f$, and the last of which occurs at time $T_l$. 
Then no non-Byzantine party receives a \code{new\_round(r+1)} notification before $T_l+\Delta_p$ and $T_l-T_f\leq \delta$.
The Pacemaker can be instantiated with any parameter $\Delta_p > 0$. 

\end{definition}

To combine the LBR and Pacemaker components in to an SMR protocol in~\sectionref{subsec:le} we fix $\Delta_p =\Delta_l$. 
Note that by using the above definition, the resulting protocol is not responsive since parties wait $\Delta_p$ before advancing rounds.
This can easily be fixed by using a more general Pacemaker definitions from~\cite{naor2019cogsworth,bravo2020making,naor2020expected}.
However, we chose the simplified version above for readability purposes since the Pacemaker is orthogonal to the thesis of our paper.

\subsection{Leader-rotation - the missing component}
\label{subsec:le}
In~\algorithmref{alg:smr} we show how to combine the LBR and Pacemaker abstractions into a leader-based SMR protocol. 
The missing component is the Leader-rotation mechanism, which exposes an $\code{choose\_leader}(r,B)$ API.
It takes a round number $r \in \mathbb{N}$ and a block $B$ and returns a party $p \in \Pi$.
The \code{choose\_leader} procedure is locally computed by each honest party at the beginning of every round. 

The Agreement property of~\algorithmref{alg:smr} follows immediately from the Agreement property of LBR, regardless of \code{choose\_leader} implementation. 
In~\appendixref{appendix} we prove that~\algorithmref{alg:smr} satisfies liveness as long as all honest parties follow the same \code{choose\_leader} procedure and that this procedure returns the same honest party at all of them infinitely often.
%
In the next section we instantiate~\algorithmref{alg:smr} with \carousel: a specific \code{choose\_leader} implementation to obtain a \smr protocol.
That is, we prove that~\algorithmref{alg:smr} with \carousel satisfies liveness, Chain-quality, and Leader-utilization.




\begin{algorithm}
\caption{Constructing SMR: code for party $p_i$}
\label{alg:smr}
\begin{algorithmic}[1]

\State $commit\_head \gets genesis$


\Upon {\code{new\_round}(r)}
    \State \emph{leader} $\gets$ \code{choose\_leader}(r, $commit\_head$) \label{l:elect}
    \State $B \gets $\emph{LBR(r,leader)} \label{l:lbr}
    
    \If{$commit\_head \longrightarrow B$}
    
        \State commit $B$ \label{l:commit}
        \Comment{all blocks in $B$'s implied chain that were not yet committed.} 
        \State $commit\_head \gets B$ \label{l:update_head}

    \EndIf

\EndUpon
\algstore{part1}

\end{algorithmic}
\end{algorithm}

\section{\carousel: A Novel Leader-Rotation Algorithm}
\label{sec:le}
In this section, we present \carousel -- our Leader-rotation mechanism.
The pseudo-code is given in~\algorithmref{alg:le}, which combined with~\algorithmref{alg:smr} allows to obtain the first \smr protocol.

We use reputation to avoid crashed leaders in crash-only executions.
Specifically, at the beginning of round $r$, an honest party checks if it has committed a block $B$ with round number $r-1$.
In this case, the set of endorsers of $B$ are guaranteed to not have crashed by round $r$.
For Chain-quality purposes, the $f$ latest authors of committed blocks are excluded from the set of endorsers, and a leader is chosen deterministically from the remaining set.


If an honest party has not committed a block with round number $r-1$, it uses a round-robin fallback scheme to elect the round $r$ leader.
Notice that different parties may or may not have committed a block with round number $r-1$ before round $r$.
In fact, the adversary has multiple ways to cause such divergence, e.g. Byzantine behavior, crashes or message delays.
As a result parties can disagree on the leader identity, and potentially compromise liveness.
We prove, however, that \carousel satisfies liveness, as well as leader utilization and Chain-quality.
Specifically, we show that (1) the number of rounds $r$ for which no honest party commits a block formed in $r$ is bounded by $O(f^2)$; and at least one honest block is committed $5f+2$ rounds. 
The argument is non-trivial, since for example, we need to show that the adversary cannot selectively alternate the fallback and reputation schemes to control the Chain-quality.

\begin{algorithm}
\caption{Leader-rotation: code for party $p_i$}
\label{alg:le}
\begin{algorithmic}[1]
\algrestore{part1}
\Procedure{$\code{choose\_leader}$}{$r, commit\_head$}

    \State $last\_authors \gets \emptyset$
    
    \If {$commit\_head.round\_number\neq r-1$ }
        \State \Return $(r\mod n)$ \label{l:return_rr}
        \Comment{round-robin fallback}
    \EndIf
    \State $active \gets commit\_head.endorsers$ \label{l:set_active}
    
    \State $block \gets commit\_head$

    \While{$last\_authors < f \wedge block \neq genesis $}
    
    
        \State $last\_authors\gets last\_authors\cup \{block.author\}$
        
        \State $block\gets block.parent$
    
    \EndWhile
    
    \State $leader\_candidates \gets active\setminus last\_authors$ \label{l:set_optional}
    
    \State return $leader\_candidates.pick\_one()$ \Comment{deterministically pick from the set} \label{l:return_rep}

\EndProcedure

\end{algorithmic}
\end{algorithm}





    
    
        
    
    
    





\subsection{Correctness}

\subsubsection{Leader-Utilization.}
In this section, we are concerned with the protocol efficiency against crash failures.
We consider time after GST, and at most $f$ parties that may crash during the execution but follow the protocol until they crash (i.e., non-Byzantine).  
We say that a party $p$ crashes in round $r$ if $r+1$ is the minimal number for which $p$ does not invoke $LBR$ in~\lineref{l:lbr}.
Accordingly, we say that a party is \emph{alive} at all rounds before it crashes.  
In addition, we say that a round $r$ occurs after GST if all \code{new\_round}(r) notifications at honest parties occur after GST.

We start by introducing an auxiliary lemma which extends the LBR Progress property for crash-only executions. 
Since in a crash-only case faulty parties follow the protocol before they crash, honest parties cannot distinguish between an honest leader and an alive leader that has not crashed yet. 
Hence, the LBR Progress property hold even if the leader crashes later in the execution.
Formal proof of the following technical lemma, using indistinguishability arguments, appears in~\appendixref{appendix}.

\begin{restatable}{lemma}{progressindis}
\label{lemma:reflecting_blocking}
In a crash-only execution, let $r$ be a round with $k \geq 2f+1$ LBR-synchronized($\ell$) invocations, such that $\ell$ is alive at round $r$, then these $k$ invocations return a certified $B$ with round number $r$ authored by $\ell$.
\end{restatable}

Furthermore, if no party crashes in a given round and the preconditions of the adapted LBR Progress conditions are met a block is committed in that round and another alive leader is chosen.

\begin{lemma}
\label{lemma:alive_leader}
If the preconditions of~\lemmaref{lemma:reflecting_blocking} hold and no party crashes in round $r$,
then $k \geq 2f+1$ honest parties commit a block for round $r$ and return the same leader $\ell'$ at~\lineref{l:elect} of round $r+1$ and $\ell'$ is alive at round $r$.


\end{lemma}
\begin{proof}

By~\lemmaref{lemma:reflecting_blocking}, $k$ honest parties return from $LBR(r,\ell)$ with a certified block $B$ with round number $r$ authored by $\ell$.
Then, since $commit\_head \longrightarrow B$, they all commit $B$ at~\lineref{l:commit} of round $r+1$.
By the LBR Reputation property, the set of $B$'s endorsers does not include parties that crashed in rounds $<r$. Since no party crashes in round $r$, $B$'s endorsers are all alive in round $r$. 
Since these $2f+1$ parties each committed block $B$ with round number $r$, in $\code{choose\_leader}$ in~\algorithmref{alg:smr}, they all use the reputation scheme (\lineref{l:return_rep}) to choose the round $r+1$ leader, that we showed is alive at round $r$.
\end{proof}

Next, we utilize the latter to prove that in a round with no crashes, it is impossible for a minority of honest parties to return with a certified block from an LBR instance. Namely, either no honest party returns a block, or at least $2f+1$ of them do.

\begin{lemma}
\label{lemma:all_or_none}

In a crash-only execution, let $r$ be a round after GST in which no party crashes. If one honest party returns from $LBR$ with a certified block $B$ with round number $r$, then $2f+1$ honest parties return with $B$.

\end{lemma}
\begin{proof}

Assume an honest party returns a certified block $B$ with round number $r$ after invoking $LBR(r,\ell)$.
By the LBR Blocking property, $\ell$ itself must have invoked $LBR(r,\ell)$ and by assumption it was \emph{alive} at round $r$.
By the LBR Endorsement property, the set of endorsing parties of $B$ contains $2f+1$ parties. Since we consider a crash-only execution, it follows by assumption that $2f+1$ party called $LBR(r,\ell)$. Due to the use of Pacemaker, these calls are LBR-synchronized($\ell$) invocations.
Finally, by~\lemmaref{lemma:reflecting_blocking} all these calls return a certified $B$ with round number $r$ authored by $\ell$.

\end{proof}

We prove that in a window of $f+2$ rounds without crashes, there must be a round with the sufficient conditions for a block to be committed for that round.

\begin{lemma}
\label{lemma:good_round}
In a crash-only execution, let $R$ be a round after GST such that no party crashes between rounds $R$ and $R+f+2$ (including).
There exists a round $R\leq r\leq R+f+2$ for which there are $2f+1$ LBR-synchronized($\ell$) invocations with a leader $\ell$ that is \emph{alive} at round $r$.
\end{lemma}
\begin{proof}
First, let us consider the $LBR$ invocations for round $R$.
By~\lemmaref{lemma:all_or_none}, if one honest party returns with a block $B$ with round number $R$, then $2f+1$ honest parties return with $B$, commit it and update $commit\_head$ accordingly (\lineref{l:update_head}).
In this case, there are $2f+1$ $\code{choose\_leader}(R+1,B)$ invocations, which all return at~\lineref{l:return_rep}.
Otherwise, no party return a block with round number $R$, and thus they all return at~\lineref{l:return_rr}.
By the code and since a block implies a unique chain, in both cases $2f+1$ honest parties return the same leader $\ell$ in $\code{choose\_leader}(R+1,B)$ (either by reputation or round-robin). 
By the Pacemaker guarantees and since $R+1$ occurs after GST, there are at least $2f+1$ LBR-synchronized($\ell$) invocations.
If $\ell$ is alive at round $R+1$, we are done.
Otherwise, $\ell$ must have been crashed before round $R$ by the alive definition and lemma assumptions.
Thus, by the LBR Blocking property no honest party commits a block for round $R$ and they all choose the same leader for the following round at~\lineref{l:return_rr}.
The lemma follows by applying the above argument for $R+f+2 - R+1 = f+1$ rounds. 

\end{proof}

Finally, we bound by $O(f^2)$ the total number of rounds in a crash-only execution for which no honest party commits a block:

\begin{lemma}
Consider a crash-only execution. After GST, the number of rounds $r$ for which no honest party commits a block formed in $r$ is bounded by $O(f^2)$. 
\end{lemma}

\begin{proof}

Consider a crash-only execution and let $R_1,R_2,\dots R_k$ the rounds after GST in which parties crash ($k\leq f$).
For ease of presentation we call a round for which no honest party commits a block formed in $r$ a \emph{skipped} round.
We prove that the number of skipped rounds between $R_i$ and $R_{i+1}$ for $1\leq i <k$ is bounded.
If $R_{i+1}-R_i<f+4$, then there are at most $f+4$ rounds and hence at most $f+4$ skipped rounds.
Otherwise, we show that at most $f+2$ rounds are skipped between rounds $R_i$ and $R_{i+1}$.

First, by~\lemmaref{lemma:good_round}, there exists a round $R_i<R_i+1\leq r\leq R_i+1+f+2<R_{i+1}$ for which there are $2f+1$ LBR-synchronized($\ell$) invocations with a leader $\ell$ that is \emph{alive} at round $r$.
By~\lemmaref{lemma:alive_leader}, since no party crashes in round $r$, $2f+1$ honest parties return the same leader $\ell'$ at~\lineref{l:elect} of round $r+1$ and $\ell'$ is alive at round $r$. 
Since no party crashes at round $r+1$ as well (because $R_{i+1}-R_i\geq f+4$), $\ell'$ is alive at round $r+1$. 
By the Pacemaker guarantees and since we consider rounds after GST, we conclude that there are at least $2f+1$ \emph{LBR-synchronized($\ell'$) invocations} for round $r+1$.
By~\lemmaref{lemma:alive_leader} applied again for round $r+1$, $2f+1$ honest parties commit a block for round $r+1$. Thus, round $r+1$ is not \emph{skipped}.
We repeat the same arguments until round $R_{i+1}$, and conclude that in each of these rounds a block is committed. Hence, the rounds that can possibly be skipped between $R_i$ and $R_{i+1}$ are $R_i \leq r' < r$.
Thus there are $O(f)$ skipped round between $R_i$ and $R_{i+1}$.
For $R_k$ we use similar arguments but since no party crashes after $R_k$, we apply~\lemmaref{lemma:alive_leader} indefinitely. We similarly conclude that there are $O(f)$ skipped rounds after $R_k$.
All in all, since $k \leq f$, we get $O(f^2)$ skipped rounds.

\end{proof}

We immediately conclude the following:

\begin{corollary}
\label{lemma:leaderutilization}
\algorithmref{alg:smr} with~\algorithmref{alg:le} satisfies Leader-utilization.
\end{corollary}



\subsubsection{Chain-Quality.}
For the purposes of the Chain-quality proof, we say that a block is committed when some honest party commits it. We say that a block $B$ with round number $r$ is \emph{immediately committed} if an honest party commits $B$ in round $r$.
When we refer to a leader elected in of~\algorithmref{alg:le} from the round-robin mechanism we mean~\lineref{l:return_rr}, and
when we refer to a leader elected from the reputation mechanism, we mean~\lineref{l:return_rep}.

We begin by showing that each round assigned with an honest round-robin leader implies a committed block in that round or the one that precedes it (not necessarily an honest block).

\begin{lemma}
\label{lemma:some_commit_per_correct}
Let $r$ be a round after GST such that $p_i =( r \mod n)$ is honest.
Then, either Byzantine block with round number $r-1$ or an honest block with round number $r-1$ or $r$ is immediately committed.
\end{lemma}
\begin{proof}
If a block is immediately committed with round number $r-1$ then we are done.
Otherwise, no honest party commits a block with round number $r-1$ in round $r-1$, and they all elect the round $r$ leader $\ell$ using the round-robin mechanism.
By the assumption, $\ell$ is honest.

By the Pacemaker, all honest invocations of $LBR(r,\ell)$ in~\lineref{l:lbr} are LBR-synchronized($\ell$).
Since there are at least $2f+1$ honest parties, by the LBR Progress property, all honest invocations return the same certified block $B$ with round number $r$ authored by $\ell$.
Then, the honest parties commit $B$ at~\lineref{l:commit}.
\end{proof}




If there are two consecutive rounds assigned with honest round-robin leaders and in addition the last $f$ committed blocks are Byzantine, then an honest block follows, as proven in the following lemma.

\begin{lemma}
\label{lemma:two_consecutive}
Let $r'$ be a round after GST such that $p_i = (r' \mod n)$ and $p_j = (r'+1 \mod n)$ are honest.
Suppose $f$ blocks with round numbers in $[r, r')$ with different Byzantine authors are committed. 
For a block $B$ with round number $r'$ or $r'+1$ that is immediately committed, there is an honest block with round number $[r, r'+1]$
on $B$'s implied chain.

\end{lemma}
\begin{proof}
By the LBR endorsement assumption and property, the author of block $B$ should be either a reputation-based, or a round-robin leader of round $r'$ or $r'+1$.
If it is a round-robin leader, then by the lemma assumption, the leader is honest and since $B$ is the head of its implied chain, the proof is complete.
Thus, in the following we assume that $B$'s author is a reputation-based leader.
By the SMR Agreement property and the lemma assumption, $B$'s implied chain contains $f$ blocks with different Byzantine authors and rounds numbers in $[r, r')$.
By the code of the reputation-based mechanism, either all $f$ byzantine authors are excluded from the \emph{leader\_candidates} which implies that $B$ has an honest author, or that there is an honest block with round number in $[r,r')$ on $B$'s implied chain.

\end{proof}

Lastly, the following lemma proves that in any window of $5f+2$ rounds an honest block is committed.

\begin{lemma}
\label{lemma:honestcommit}
Let $r$ be a round after GST.
At least one honest block is committed with a round number in $[r, r+5f+2]$.
\end{lemma}
\begin{proof}
Suppose for contradiction that no honest block with round number in $[r, r+5f+2]$ is committed.
There are at least $f$ rounds $r'$ in $[r, r+3f+1)$, such that rounds $r'-1$ and $r'$ are allocated an honest leader by the round-robin mechanism.
By~\lemmaref{lemma:some_commit_per_correct}, a block with round number $r'-1$ or $r'$ is immediately committed.
Due to~\lemmaref{lemma:some_commit_per_correct} and the contradiction assumption, for any such round $r'$, a Byzantine block with round number $r'-1$ is immediately committed.
Since $r'-1$ has an honest round-robin leader, the block must be committed from the reputation mechanism.

It follows that $f$ Byzantine blocks with round numbers in $[r, r+3f+1)$ are immediately committed from the reputation mechanism, and consequently, they all must have different authors.
Note that there exists $r' \in [r+3f+1,r+5f+2)$ (in a window of $2f+1$ rounds), such that the round-robin mechanism allocates honest leaders to rounds $r'$ and $r'+1$.
By~\lemmaref{lemma:some_commit_per_correct}, a block $B$ with round number $r'$ or $r'+1$ is immediately committed.~\lemmaref{lemma:two_consecutive} concludes the proof. 
\end{proof}

We conclude the following:

\begin{corollary}
\label{lemma:chainquality}
\algorithmref{alg:smr} with~\algorithmref{alg:le} satisfies Chain-quality and Liveness.
\end{corollary}




Taken jointly,~\Cref{lemma:leaderutilization}, \Cref{lemma:chainquality}, and the Agreement property proved in~\Cref{subsec:le} yield the
following theorem:

\begin{theorem}
\algorithmref{alg:smr} with~\algorithmref{alg:le} implements \smr.
\end{theorem}

\section{Implementation}
We implement \carousel on top of a high-performance open-source implementation of HotStuff\footnote{
\url{https://github.com/asonnino/hotstuff}
}~\cite{yin2019hotstuff}. We selected this implementation because it implements a Pacemaker~\cite{yin2019hotstuff}, contrarily to the implementation used in the original HotStuff paper\footnote{\url{https://github.com/hot-stuff/libhotstuff}}.
Additionally, it provides well-documented benchmarking scripts to measure performance in various conditions, and it is close to a production system (it provides real networking, cryptography, and persistent storage). 
It is implemented in Rust, uses Tokio\footnote{\url{https://tokio.rs}} for asynchronous networking, ed25519-dalek\footnote{\url{https://github.com/dalek-cryptography/ed25519-dalek}} for elliptic curve based  signatures, and data-structures are persisted using RocksDB\footnote{\url{https://rocksdb.org}}. It uses TCP to achieve reliable point-to-point channels, necessary to correctly implement the distributed system abstractions.
By default, this HotStuff implementation uses traditional round-robin to elect leaders; we modify its \texttt{LeaderElector} module to use \carousel instead. Implementing our mechanism requires to add less than 200 LOC, and does not require any extra protocol message or cryptographic tool.
We are open-sourcing \carousel\footnote{
\ifdefined\cameraReady
\url{https://github.com/asonnino/hotstuff/tree/leader-reputation}
\else
Link omitted for blind review.
\fi
} along with any measurements data to enable reproducible results\footnote{
\ifdefined\cameraReady
\url{https://github.com/asonnino/hotstuff/tree/leader-reputation/data}
\else
Link omitted for blind review.
\fi
}.
\section{Evaluation}
%
We evaluate the throughput and latency of HotStuff equipped \carousel through experiments on Amazon Web Services (AWS). We then show how it improves over the baseline round-robin leader-rotation mechanism.
We particularly aim to demonstrate that \carousel (i) introduces no noticeable performance overhead when the protocol runs in ideal conditions (that is, all parties are honest) and with small committees, and (ii) drastically improves both latency and throughput in the presence of crash-faults. Note that evaluating BFT protocols in the presence of Byzantine faults is still an open research question~\cite{twins}.

We deploy a testbed on AWS, using \texttt{m5.8xlarge} instances across 5 different AWS regions: N. Virginia (us-east-1), N. California (us-west-1), Sydney (ap-southeast-2), Stockholm (eu-north-1), and Tokyo (ap-northeast-1). Parties are distributed across those regions as equally as possible. Each machine provides 10Gbps of bandwidth, 32 virtual CPUs (16 physical core) on a 2.5GHz, Intel Xeon Platinum 8175, 128GB memory, and run Linux Ubuntu server 20.04. 

In the following sections, each measurement in the graphs is the average of 5 independent runs, and the error bars represent one standard deviation. 
Our baseline experiment parameters are: 10 honest parties, a block size of 500KB, a transaction size of 512B, and one benchmark client per party submitting transactions at a fixed rate for a duration of 5 minutes. We then crash and vary the number of parties through our experiments to illustrate their impact on performance. The leader timeout value is set to 5 seconds for committees of 10 and 20, and increased to 10 seconds for committees of 50. When referring to \emph{latency}, we mean the time elapsed from when the client submits the transaction to when the transaction is committed by one party. We measure it by tracking sample transactions throughout the system.

\subsection{Benchmark in Ideal Conditions}
\figureref{fig:happy-path} depicts the performance of HotStuff with both \carousel and the baseline round-robin running with 10, 20 and 50 honest parties. For small committees (10 parties), the performance of the baseline round-robin HotStuff is similar to HotStuff equipped with \carousel. We observe a peak throughput around 70,000 tx/s with a latency of around 2 seconds. This illustrates that the extra code required to implement \carousel has negligible overhead and does not degrade performance when the committee is small.
When increasing the committee (to 20 and 50 parties), HotStuff with \carousel greatly outperforms the baseline: the bigger the committee, the bigger the performance improvement. With 50 nodes, the throughput of our mechanism based HotStuff increases by over 2x with respect to the baseline, and remains comparable to the 10-parties testbed. After a few initial timeouts, \carousel has the benefit to focus on electing performant leaders. Leaders on more remote geo-locations that are typically slower are elected less often, the protocol is thus driven by the most performant parties. Latency is similar for both implementations and around 2-3 seconds.

\begin{figure*}[t]
\vspace{-0.4cm}
\centering
\includegraphics[width=\textwidth]{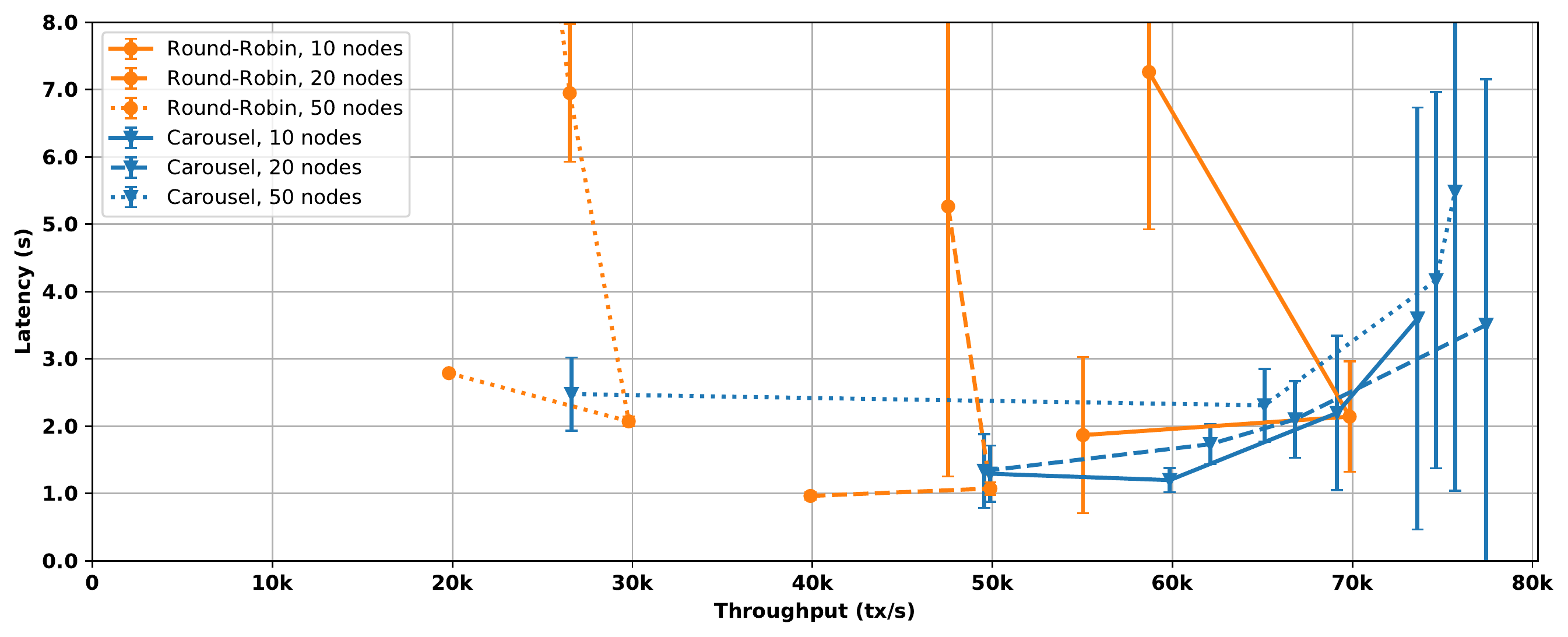}
\vspace{-0.7cm}
\caption{
Comparative throughput-latency performance of HotStuff equipped with \carousel and with the baseline round-robin. WAN measurements with 10, 20, 50 parties. No party faults, 500KB maximum block size and 512B transaction size.
}
\label{fig:happy-path}
\vspace{-0.3cm}
\end{figure*}

\subsection{Performance under Faults}
\figureref{fig:dead-nodes} depicts the performance of HotStuff with both \carousel and the baseline round-robin when a set of 10 parties suffers 1 or 3 crash-faults (the maximum that can be tolerated). The baseline round-robin HotStuff suffers a massive degradation in throughput as well as a dramatic increase in latency. For three faults, the throughput of the baseline HotStuff drops over 30x and its latency increases 5x compared to no faults.
In contrast, HotStuff equipped with \carousel maintains a good level of throughput: our mechanism does not elect crashed leaders, the protocol continues to operate electing leaders from the remaining active parties and is not overly affected by the faulty ones.
The reduction in throughput is in great part due to losing the capacity of faulty parties. 
When operating with 3 faults, \carousel provides a 20x throughput increase and about 5x latency reduction with respect to the baseline round-robin.

\begin{figure*}[t]
\vspace{-0.4cm}
\centering
\includegraphics[width=\textwidth]{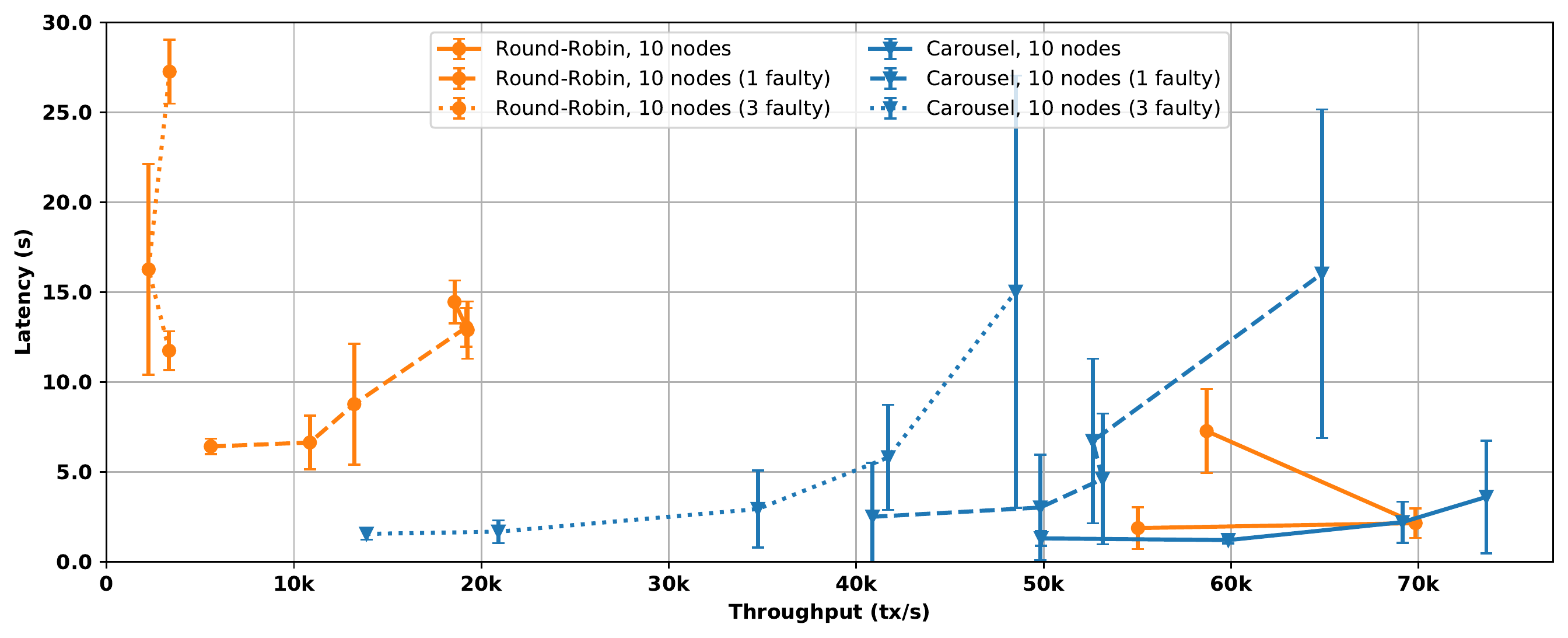}
\vspace{-0.7cm}
\caption{
Comparative throughput-latency performance of HotStuff equipped with \carousel and with the baseline round-robin. WAN measurements with 10 parties. Zero, one and three party faults, 500KB maximum block size and 512B transaction size.
}
\label{fig:dead-nodes}
\vspace{-0.3cm}
\end{figure*}


\figureref{fig:time} depicts the evolution of the performance of HotStuff with both \carousel and the baseline round-robin when gradually crashing nodes through time. For roughly the first minute, all parties are honest; we then crash 1 party (roughly) every minute until a maximum of 3 parties are crashed. The input transaction rate is fixed to 10,000 tx/s throughout the experiment. Each data point is the average over intervals of 10 seconds.
For roughly the first minute (when all parties are honest), both systems perform ideally, timely committing all input transactions. Then, as expected, the baseline round-robin HotStuff suffers from temporary throughput losses when a crashed leader is elected. Similarly, its latency increases with the number of faulty parties, and presents periods where no transactions are committed at all.
In contrast, HotStuff equipped with \carousel delivers a stable throughput by quickly detecting an eliminating crashed leaders. Its latency is barely affected by the faulty parties.
This graph clearly illustrates how \carousel allows HotStuff to deliver a seamless client experience even in the presence of faults.

\begin{figure*}[t]
\vspace{-0.4cm}
\centering
\includegraphics[width=\textwidth]{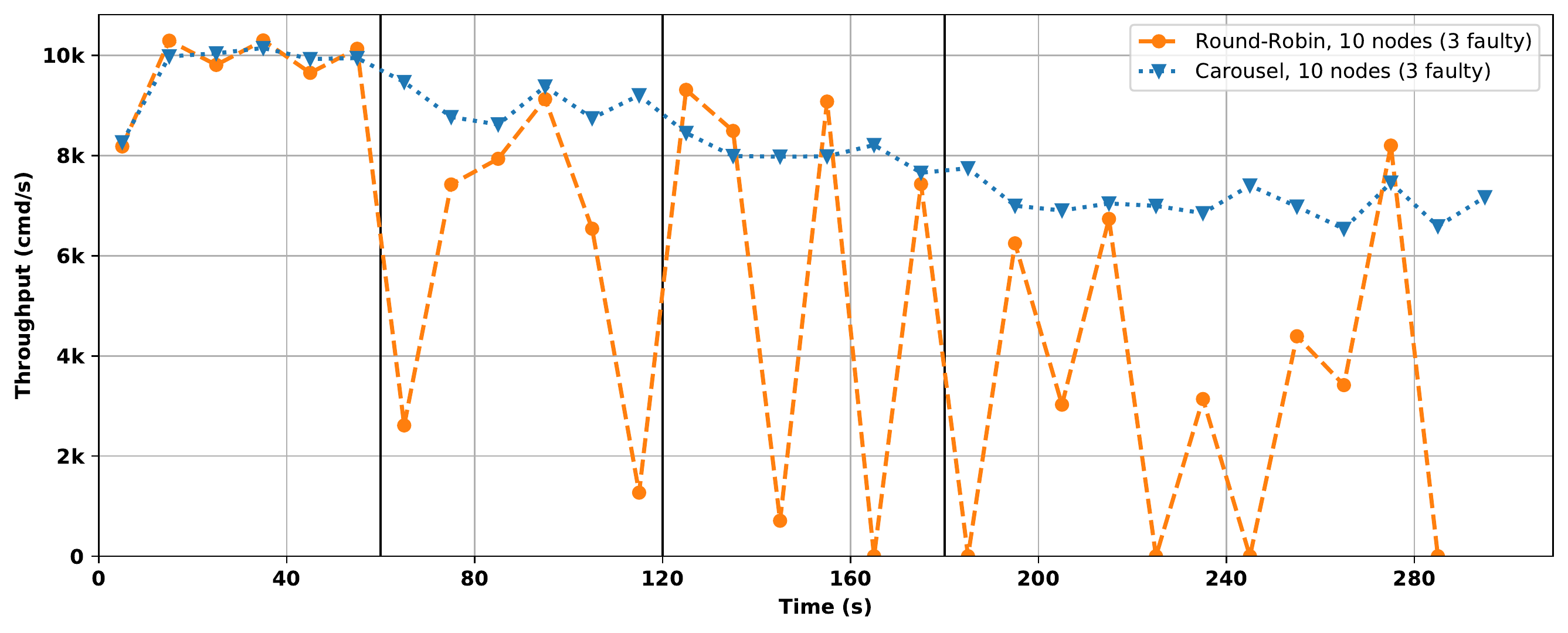}
\includegraphics[width=\textwidth]{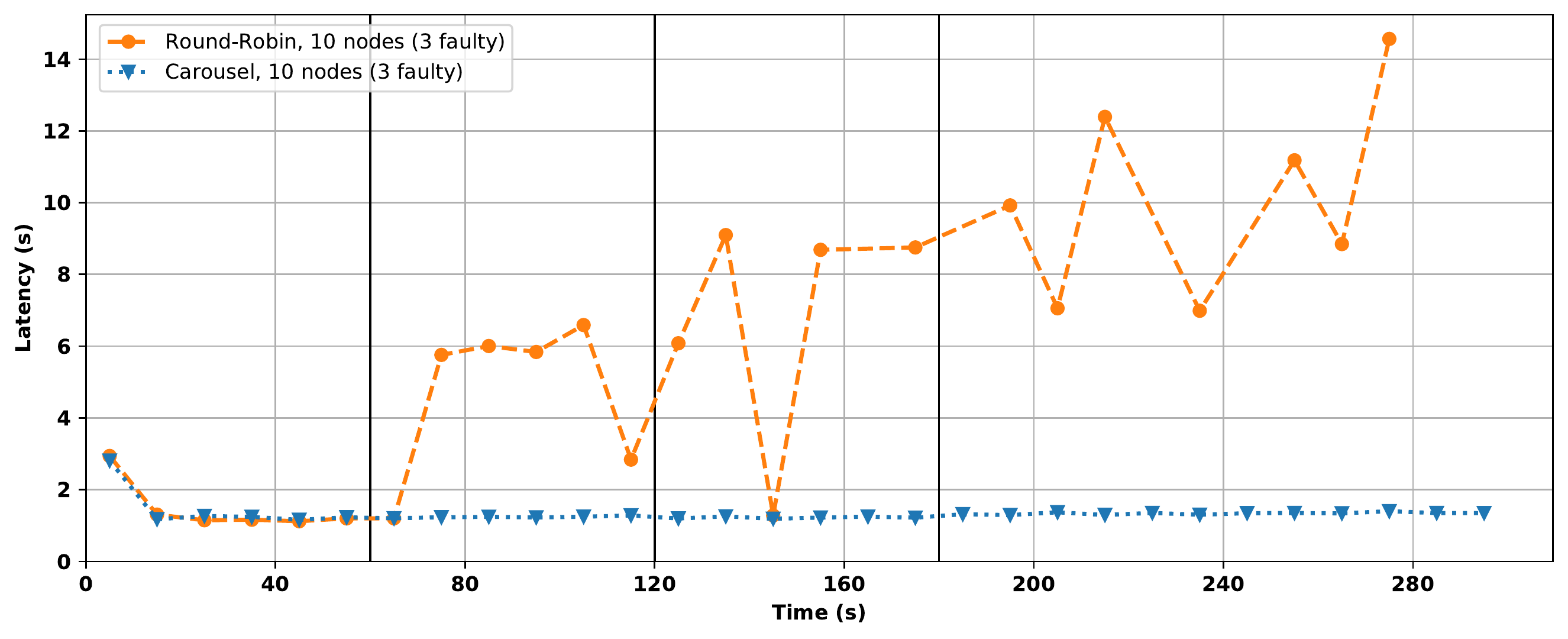}
\vspace{-0.7cm}
\caption{
Comparative performance of HotStuff equipped with \carousel and with the baseline round-robin when gradually crashing nodes through time. The input transactions rate is fixed to 10,000 tx/s; 1 party (up to a maximum of 3) crashes roughly every minute. WAN measurements with 10 parties, 500KB maximum block size and 512B transaction size.
}
\label{fig:time}
\vspace{-0.3cm}
\end{figure*}

\section{Conclusions}

Leader-rotations mechanisms in chaining-based SMR protocols were previously overlooked. 
Existing approaches degraded performance by keep electing faulty leaders in crash-only executions.
We captured the practical requirement of leader-rotation mechanism via a Leader-utilization property, use it define the \smr problem, and described an algorithm that implements it.
That is, we presented a locally executed algorithm to rotate leaders that achieves both: Leader-utilization in crash-only executions and Chain-quality in Byzantine ones.
We evaluated our mechanism in a Hotstuff-based open source system and demonstrated drastic performance improvements in both throughput and latency compared to the round-robin baseline.


\newpage






\bibliographystyle{plain}
\bibliography{references}


\begin{subappendices}
\renewcommand{\thesection}{\Alph{section}}%
\newpage
\section{Correctness}

\label{appendix}

\begin{lemma}
\label{lemma:smr_liveness}
If \code{choose\_leader} returns the same honest party at all honest parties for infinitely many rounds, then each honest party commits an unbounded number of blocks.
\end{lemma}
\begin{proof}
If \code{choose\_leader} returns the same honest party at all honest parties for infinitely many rounds, then there are infinitely many rounds after GST for which it does so.
Let $r$ be such a round.
By the Pacemaker guarantees, all honest parties make LBR-synchronized($\ell$) invocations with the same honest leader $\ell$ returned from the \code{choose\_leader} procedure.
By the LBR Progress property, they all return a certified block $B$ and commit it at~\lineref{l:commit}.
\end{proof}

\progressindis*
\begin{proof}
Let $\pi_1$ be a crash-only execution, 
  such that round $r$ has $k \geq 2f+1$ \emph{LBR-synchronized($\ell$) invocations} with a leader $\ell$ that is alive at round $r$. 
If $\ell$ is honest, then the LBR Progress property concludes the proof.

Otherwise, $\ell$ is faulty and by definition it crashes in round $> r$.
Let $\pi_2$ be a crash-only execution that is identical to $\pi_1$ until $\ell$ crashes, and the rest of $\pi_2$ is an arbitrary execution where the honest parties in $\pi_1$ remain honest but $\ell$ never crashes and is also honest.
Thus, in $\pi_2$ the preconditions of the LBR Progress property hold and all $k$ \emph{LBR-synchronized($\ell$) invocations} return a certified $B$ with round number $r$ authored by $\ell$.

An $LBR(r, \ell)$ invocation by any party $p$ completes within $\Delta_l$ time, and starts immediately after Pacemaker's $\code{new\_round(r)}$ notification at $p$ (because \code{choose\_leader} is computed locally and takes $0$ time).
By Pacemaker's guarantees, no party receives $\code{new\_round}(r+1)$ notification until $\Delta_p = \Delta_l$ time after the last $\code{new\_round}(r+1)$ notification at some party, hence all $LBR(r, \ell)$ invocations must complete before any party receives a $\code{new\_round}(r+1)$ notification.

$\pi_1$ and $\pi_2$ are identical until $\ell$ crashes, which must happen after $\ell$ receives its $\code{new\_round}(r+1)$ notification from the Pacemaker.
This is because $\ell$ is alive in round $r$ and follows the protocol, invoking $LBR$ in round $r+1$ after receiving the $\code{new\_round}(r+1)$ notification.
As a result, $\pi_1$ and $\pi_2$ are indistinguishable to all $LBR(r, \ell)$ invocations, and the $k$ \emph{LBR-synchronized($\ell$) invocations} in $\pi_1$
return certified block $B$ with round number $r$ authored by $\ell$ as in $\pi_2$, as desired.
\end{proof}
\end{subappendices}

\end{document}